\documentclass[letterpaper,11pt,colorinlistoftodos]{article}
\usepackage{latexsym,array,delarray,amsthm,amssymb,amsmath,epsfig}
\usepackage[hidelinks]{hyperref}
\usepackage[usenames]{color}
\usepackage{caption}
\usepackage{subcaption}
\usepackage{verbatim}
\usepackage{todonotes}

\theoremstyle{plain}
\newtheorem{thm}{Theorem}[section]
\newtheorem{lemma}[thm]{Lemma}
\newtheorem{prop}[thm]{Proposition}
\newtheorem{cor}[thm]{Corollary}
\newtheorem{conj}[thm]{Conjecture}
\newtheorem*{thm*}{Theorem}
\newtheorem*{lemma*}{Lemma}
\newtheorem*{prop*}{Proposition}
\newtheorem*{cor*}{Corollary}
\newtheorem*{conj*}{Conjecture}

\theoremstyle{definition}
\newtheorem{defn}[thm]{Definition}
\newtheorem{ex}[thm]{Example}

\theoremstyle{remark}
\newtheorem*{rmk}{Remark}

\newcommand{\PP}{\mathbb{P}}

\newcommand{\RR}{\mathbb{R}}

\newcommand{\co}{\mathcal{O}}

\newcommand{\Cof}{\operatorname{Cof}}
\newcommand{\desc}{\operatorname{desc}}
\newcommand{\MRCA}{\operatorname{MRCA}}

\newcommand{\UE}{UE}

\newcommand{\EE}{EE}

\newcommand{\CK}{CK}
\newcommand{\Flat}{\operatorname{Flat}}

\newcommand{\tc}{\text{:}} 

\usepackage[normalem]{ulem}

\newcommand\blfootnote[1]{%
  \begingroup
  \renewcommand\thefootnote{}\footnote{#1}%
  \addtocounter{footnote}{-1}%
  \endgroup
}

\begin{document}

\title{Phylogenomic Models from Tree Symmetries}

\author{Elizabeth S.~Allman, Colby Long, and John A.~Rhodes}
\maketitle

\abstract{A model of  genomic sequence evolution on a species tree should include
not only a sequence substitution process, but also a coalescent process,
since different sites may evolve on different gene trees due to
incomplete lineage sorting. Chifman and Kubatko initiated the study of
such models, leading to the development of the SVDquartets methods of
species tree inference. A key observation was that symmetries in an
ultrametric species tree led to symmetries in the joint distribution of
bases at the taxa. In this work, we explore the implications of such
symmetry more fully, defining new models incorporating only the
symmetries of this distribution, regardless of the mechanism that might
have produced them. The models are thus supermodels of many standard
ones with mechanistic parameterizations. We study phylogenetic
invariants for the models, and establish identifiability of species tree
topologies using them. 
\blfootnote{\\
\textbf{Funding:} ESA and JAR were partially supported by NSF grant 2051760 and NIH grant P20GM103395. \\}
}

\section {Introduction}

The SVDquartets method of Chifman and Kubatko \cite{CK14,CK15} initiated a novel framework for species tree inference from genomic-scale data. 
Recognizing that individual sites may evolve along different ``gene trees" due to the population-genetic effect of incomplete lineage sorting, 
their method is designed to work with site pattern data generated by the multispecies coalescent model of this process combined with a standard model of site-substitution. 
However, rather than try to associate particular gene trees to sites,
they regard the observed site pattern distribution as a \emph{coalescent mixture}.
This effectively integrates the individual gene trees out of the analysis
and allows them to formulate statistical tests based on an algebraic understanding of the site pattern frequencies.
These tests detect the unrooted species tree topology in the case of four taxa. 
For a larger set of taxa, species trees can be found by inferring each quartet and then applying some 
method of quartet amalgamation. This leads to their SVDquartets method of species tree inference, which is implemented in PAUP* \cite{swofford2016} and which continues to be an important tool for practical phylogenetic inference (e.g., \cite{Jebb2020, Razifard2020, Crowl2022}).

The inference of unrooted 4-taxon species tree topologies in the SVDquartets approach is based on 
an algebraic insight that a certain flattening matrix built from  the site pattern distribution should have low rank on a 
distribution exactly arising from the model. The mathematical arguments for this in \cite{CK15} are based on the existence of a rooted cherry (i.e., a 2-clade) on an  
ultrametric species tree, leading to a symmetry in the site pattern distribution. Since any rooted 4-taxon tree with unrooted 
topology $ab|cd$ must display at least one of the clades $\{a,b\}$ or $\{c,d\}$, detecting that one or both of these clades
 is present is equivalent to determining the unrooted tree. 
 The SVDquartets method tests precisely this, without determining which of the clades is present.

In this work, we examine the algebraic framework underlying the work of Chifman and Kubatko and its subsequent extensions. 
We observe that the symmetry conditions implied by the  Chifman-Kubatko model are key to their inference approach. 
Based on this observation, we formulate several statistical models, 
encompassing those of \cite{CK15} as well as several more general mechanistic models, which capture the fundamental assumptions needed to justify SVDquartets.
In contrast with the sorts of models generally used in empirical phylogenetics, 
which have a mechanistic interpretation
(e.g., generation of gene trees by the coalescent process, generation of sequences by 
site-substitution models on the gene trees), the models here have only a descriptive interpretation, as they are defined algebraically by constraints on site pattern distributions. 

One consequence of defining our models in this way is that it becomes more clear that 
SVDquartets can give consistent species tree inference for mechanistic mixture models more general than that described in \cite{CK15} (as hinted by results in \cite{ LK18geneflow, LK18}). 
In fact, it is easy to formulate plausible mechanistic models with many parameters 
(e.g. mixtures with many different base substitution processes) for which many of the numerical parameters must be non-identifiable, but for which SVDquartets inference of the species tree topology is statistically 
consistent. 
Such generality can be viewed as a strength of SVDquartets, as model misspecification arising from assumption of a simple substitution process across the entire genome is avoided.

A second consequence is that our models highlight a symmetry in the site pattern distribution that reflects the \emph{rooted} 
species tree, a symmetry that is present even for 3-taxon trees.
Methods for inference of the species tree root in the same framework were proposed 
in \cite{GK16, TK17}, but both of these works considered four taxa at a time, which is the smallest \emph{unrooted} tree size in which topologies may differ.
Since rooted trees are determined by their rooted triples, focusing on the 3-taxon 
case offers clear advantages for developing new inference methods. Unfortunately, 
in doing so, we lose the ability to naturally base statistical inference on rank conditions on 
matrices of the sort that underlie SVDquartets. 
Indeed, the possible flattening matrices for DNA site pattern data from the 
Chifman-Kubatko model in the 3-taxon case are all $4\times 16$ with full rank, so rank alone cannot distinguish them. 
As a consequence, the matrix Singular Value Decomposition (SVD) of the flattening matrix, which is used to determine approximate rank in the 
SVDquartets method, has no obvious role. However, we present an alternative matrix that must satisfy certain rank conditions
in the 3-taxon case, which suggests it may be possible to develop a 3-taxon method analagous to SVDquartets.

Our work here is theoretical, dealing primarily with model definitions and algebraic consequences of those models. 
We suggest its implications for data analysis, but do not explore 
possible methods based on these results in depth.
We begin the next section with a 
review of the model of
\cite{CK15} and use it to motivate the introduction of our first model, the \emph{ultrametric exchangeable model}. 
We then discuss  a 
number of its submodels on ultrametric trees, and 
show in Section \ref{sec:UEid} that the species tree parameter of these models is
generically identifiable and species tree inference by SVDquartets is justified for all. 
In Section \ref{sec:UEdim}, we give a recursive formula for computing the dimension 
of the ultrametric exchangeable model, in terms 
of the dimensions of its subtree models joined at the root. 
This indicates that the dimension depends on the topology of the tree, 
which has implications for inference methods. 
In Section \ref{sec:genmodel}, we drop the assumption of an ultrametric species tree, 
reviewing the model of \cite{LK18} in this setting
and using it  to motivate our second model, the \emph{extended exchangeable model}. In Section \ref{sec:EE3taxa}, we explore the extended exchangeable model in
more depth by restricting to 3-taxon trees and determining several algebraic invariants
of this model. Finally, in Section \ref{sec:EEid}, we show that
the species tree parameter of the extended exchangeable model, as well as 
those of several mechanistic models that it contains, are generically identifiable.

\section{A genomic model of site patterns on ultrametric trees}\label{sec:UE}

We begin by reviewing the simplest mechanistic model of Chifman and Kubatko \cite{CK15}.   For emphasis, we call this model
(and others) \emph{mechanistic} since it incorporates models of both incomplete lineage sorting and of site substitution (e.g., GTR) in its formulation.
Many mechanistic models, 
including that of \cite{CK15}, will be included as submodels of the more general non-mechanistic models
we define below and for which the theory underlying SVDquartets applies more broadly.

Specifically, let $\sigma^+=(\psi^+,\lambda)$ be an ultrametric rooted species tree on a set of taxa $X$, with rooted leaf-labelled topology $\psi^+$ 
and edge lengths $\lambda$ in number of generations. Let $N$ be a single constant population size for all populations (i.e., edges) 
in the species tree, and $\mu$ a single scalar mutation rate for all populations. For a DNA substitution model fix some GTR rate 
matrix $Q$ with associated stable base distribution $\pi$. 

These parameters determine a DNA site pattern distribution as follows: a site is first assigned a leaf-labelled ultrametric gene tree $T$ sampled 
under the multispecies coalescent model on $\sigma^+$ with populations size $N$, with one gene lineage sampled per taxon. Then a site 
evolves on $T$ according to the base substitution model with root distribution $\pi$ and rate matrix $\mu Q$. Site patterns thus have a distribution 
which is a \emph{coalescent independent mixture} of site pattern distributions arising from the same GTR model on individual gene trees. We denote 
this model by $\text {CK}=\text {CK}(\sigma^+,N,\mu,Q,\pi)$. 
(While \CK\ has a mild non-identifiability issue in that $\lambda$, $N$,  and $\mu$ 
are not separately identifiable, this will not 
be of concern in this work since our focus is on inferring the topology $\psi^+$.)

\smallskip

A key feature of the \CK\ model is an exchangeability property that it inherits from the multispecies coalescent, due to the nature of the substitution model. 
Specifically, suppose $\{a,b\}\subseteq X$ is a  2-clade displayed on $\sigma^+$. Then for any metric gene tree $T$, let $T'$ be the gene tree obtained 
from $T$ by switching the labels $a$ and $b$. Then the ultrametricity of  $\sigma^+$ together with exchangeability of lineages under the coalescent 
model implies $T$ and $T'$ are equiprobable. Now consider any site pattern $z=(z_1,z_2,\dots, z_n)$ for $X$, where $z_i\in\{A,G,C,T\}$ is the base 
for taxon $x_i\in X$, and let $z'$ be the site pattern with the $a$ and $b$ entries interchanged. Then under the base substitution model the 
probability of $z$ on $T$ equals the probability of $z'$ on $T'$.
Thus, with $\mathcal T$ denoting the space of all metric gene trees $T$ on $X$,
\begin{align}
\PP(z \mid \sigma^+,N,\mu,Q,\pi)&=\int_{\mathcal T} \PP(z\mid T,\mu,Q,\pi) \PP(T\mid \sigma^+,N)\,dT\notag \\
&=\int_{\mathcal T} \PP(z'\mid T',\mu,Q,\pi) \PP(T'\mid \sigma^+,N)\,dT'\label{eq:exch}\\
&=\PP(z' \mid \sigma^+,N,\mu,Q,\pi).\notag\end{align}

Thus any 2-clade  on the species tree produces symmetry in the site pattern frequency distribution. 
Moreover, since both the multispecies coalescent model and the sequence substitution model are well behaved with respect to 
marginalizing over taxa, it immediately follows that 2-clades on the induced subtrees $\sigma^+|_Y$ on subsets $Y\subset X$ 
will produce symmetries in the marginalizations of the site pattern distribution to $Y$.

This motivates the following definition of an algebraic model of site pattern probabilities. 
In this definition and in what follows, it will be convenient to regard a site pattern probability distribution $P$
from a $\kappa$-state model on an $n$-leaf tree as an \emph{$n$-way site pattern probability tensor}.
That is, we regard $P=(p_{i_1 \dots i_n})$ as a $\kappa\times \cdots \times\kappa$ array 
with non-negative entries adding to 1, where $p_{i_1 \dots i_n}$ denotes the probability that the $n$ (ordered) taxa are
in state $(i_1, \dots, i_n)$.

\begin{defn}
\label{defn: uemodel}
Let $\psi^+$ be a rooted binary topological species tree on $X$, and $\kappa\ge2$. Then the $\kappa$-state \emph{ultrametric exchangeable model}, 
\UE$_\kappa(\psi^+$), is the set of all $|X|$-way site pattern probability tensors $P$, 
such that for every $Y\subseteq X$, and 
every 2-clade $\{a,b\}$ on $\psi^+|_Y$, the marginal distribution $P_Y$ of site 
patterns on $Y$ is invariant under exchanging the $a$ and $b$ indices.
The collection of all distributions as $\psi^+$ ranges over rooted 
binary topological trees on $X$ is the UE model (or the UE$_\kappa$ model to avoid ambiguity).
\end{defn}

Although this model has `ultrametric' in its name, note that the tree $\psi^+$ is a topological rooted tree, with no edge lengths. 
`Ultrametric' here refers to the motivation for the model, generalizing the \CK\ model on an ultrametric species tree discussed above. 
While one can contrive mechanistic models on non-ultrametric trees that lead to distributions in the \UE$_\kappa(\psi^+$) 
model, we do not find them very natural, and prefer to highlight the ultrametricity that a plausible mechanistic model is 
likely to require to lie within \UE$_\kappa(\psi^+$).

It is important to note that unlike most models in phylogenetics, including the \CK\ model  above, the \UE\ model is not defined through 
mechanistically-interpretable parameters. Rather it has a descriptive form relating entries of the model's joint distributions, 
chosen to reflect certain implicit features of the \CK\ model.  The \UE\ model then can be viewed as a relaxation, or supermodel, of 
that more restrictive model.    

\medskip

\begin{ex}
\label{ex:UEmodel}
Let $\psi^+$ be the rooted 3-taxon tree $(a,(b,c))$ and consider
a 2-state substitution model with states $\{0,1\}$. A probability
distribution for the \UE$\big(\, (a,(b,c))\,  \big)$ model is $P = (p_{ijk})$, a
$2 \times 2 \times 2$ array with entries the joint probabilities for assignments
of states to the taxa, 
$p_{ijk} = \PP (a = `i\text{'}, b =`j\text{'}, c=`k\text{'})$.

Since the constraints on the model arise only from subsets $Y\subseteq\{a,b,c\}$ 
that contain at least two taxa, there are four subsets of interest:
$$\{a,b,c\},\ \{b,c\}, \ \{a,b\},\ \{a,c\}.$$
Then \UE$_2(\psi^+)$ is a subset of the probability simplex $\Delta^7\subset \mathbb{R}^8$
defined by the following linear equations.
\begin{align*}
	\{a,b,c\}&:\  \begin{cases} p_{010} =   p_{001} \\
	 p_{101} =   p_{110} \end{cases}\\
	 \{b,c\}&:\  p_{001} + p_{101} =   p_{010} + p_{110} \\
	  \{a,b\}&:\  p_{010} + p_{011} =   p_{100} + p_{101} \\
	  \{a,c\}&:\  p_{001} + p_{011} =   p_{100} + p_{110} \\
\end{align*}
The first two constraints, for  $\{a,b,c\}$, express that slices
on the first index of probability tensors in \UE$_2(\psi^+)$ are symmetric.  Specifically, 
if $P_{z \cdot \cdot }$ denotes the conditional distribution of $b,c$
when $a$ is in state $z$, then the $2\times 2$ matrix $P_{z \cdot \cdot }$ is symmetric for each $z \in \{0,1\}$.
These imply the third equation, for $\{b,c\}$, expressing that marginalizing over the first index gives  a symmetric matrix.
The fourth equation, for $\{a,b\}$, is independent of the first three, but with them implies the fifth one, for $\{a,c\}$.

Taking into account the probabilistic requirement that $\sum_{i,j,k\in\{0,1\}}p_{ijk}=1,$ we see
the model is a restriction of a 4-dimensional affine space to the simplex $\Delta^7$ 
with $0\le p_{ijk}\le 1$.
\end{ex}

It is clear that far more complicated models of site pattern evolution
on a species tree than the \CK\ model give rise to distributions which also lie within the \UE\ model, since
the only requirement is that the resulting site pattern distributions reflect the symmetries of the species tree. 
For instance, in \cite{CK15}, an extension is given to allow for $\Gamma$-distributed rate variation across sites.
A further generalization, allowing for edge-dependent
variation of the population size $N = N_e$, as well as time-dependent variation in the mutation rate $\mu$ 
across the species tree, can also easily be seen to produce distributions lying within \UE.
Since the symmetry conditions arising from the species tree are  linear constraints on the site pattern probability distributions, 
arbitrary mixtures of models exhibiting the same symmetries will again exhibit these symmetries.
Thus, the mechanistic models in \cite{ALR18} on ultrametric trees that allow for 
variation in the substitution rate matrix across sites also are submodels of \UE.
Similarly, it has been shown that a model of gene flow on a 3-taxon ultrametric species tree 
will produce site pattern probability distributions that reflect the symmetry in the 2-clade of the species
tree  \cite[Proposition 0.8]{LK18geneflow}. In focusing on the \UE\ model we obtain results that apply to 
all these models, and possibly more to be formulated in the future.

\section{Generic identifiability of trees under the \UE\ model}\label{sec:UEid}

To use a statistical model for valid inference, it is necessary that any 
parameter one wishes to infer be \emph{identifiable}; that is,
a probability distribution from the model must uniquely determine the parameter. 
For phylogenetic models, this strict notion is generally too strong to hold, but one can often 
establish a similar generic result, that the set of distributions on which identifiability fails is of 
negligible size (measure zero) within the model. The following theorem is in this vein.

\begin{thm} The rooted binary topological tree $\psi^+$ is identifiable 
from a generic probability distribution in the \UE\ model.
\end{thm} 
\begin{proof} Fix $\kappa$ and a taxon set $X$. Since for each binary species tree topology
$\psi^+$ the symmetry conditions are expressible by linear equations, the  \UE\ model 
for $\psi^+$ is the intersection of a linear space with the probability simplex. 
We establish the result by showing that the linear model spaces for different $\psi^+$ are not 
contained in one another, since then their intersection is of lower dimension and hence of 
measure zero within them.

That the linear spaces are not contained in one another will follow by establishing that for each 
$\psi^+$ there is at least one distribution in
 \UE$_\kappa(\psi^+)$ that fails to have any `extra symmetry' required for it to be in the model for a 
 different tree.   To construct such a distribution, assign positive edge lengths to $\psi^+$ so that the tree is ultrametric, and 
 consider on it the $\kappa$-state analog of the (non-coalescent) Jukes-Cantor (henceforth denoted JC) model. The resulting site pattern distribution $P$ is easily seen to have the necessary symmetries to lie in the UE model.

To show $P$ has no extra symmetries, suppose to the contrary that there is a $Y\subset X$ containing two taxa $a,c$ where 
$P|_Y$ is invariant under exchanging the $a$ and $c$ indices, yet $a,c$ do not form a cherry on $\psi^+|_Y$. Then, after 
possibly interchanging the names of $a,c$, there is a third taxon $b$ such that the rooted triple $((a,b),c)$ is displayed 
on $\psi^+|_Y$. Moreover, by further marginalizing to $Y'=\{a,b,c\}$, we have that $P|_{Y'}$ arises from a 
Jukes-Cantor model on a 3-taxon ultrametric tree with positive edge lengths and rooted topology $((a,b),c)$, and exhibits $a,c$ symmetry.

To see that this is impossible, note that if $P|_{Y'}$ has both $a,b$ and $a,c$ symmetry, then it also exhibits 
$b,c$ symmetry. Thus, all marginalizations of $P|_{Y'}$ to two taxa are equal. 
This implies all JC distances between taxa, which can be computed from these marginalizations, are equal. 
This contradicts that the tree was binary.
\end{proof}

Note that the proof above did not consider a coalescent process in any way in order to show that
extra symmetries do not generically hold in \UE $(\psi^+)$.  However, since applications may consider submodels of the UE model, such as the CK model, it is necessary to ensure they do not lie within the exceptional set of non-generic points in the UE model where tree identifiability may fail. To address this issue, we seek an 
identifiability result for more 
general mechanistic models that have an \emph{analytic parameterization}, by which we mean that for each topology $\psi^+$ there is an analytic map from a full-dimensional connected subset of $\mathbb R^k$, for some $k$, to the set of 
probability distributions comprising the model.  
For example, if $\sigma^+$ is a rooted metric species tree with shape
$\psi^+$,   and site pattern frequency distributions are generated on gene trees arising under the coalescent using the
GTR+I+$\Gamma$ model, then the collection of such distributions is given by an analytic parameterization, 
and as such is a submodel of \UE$(\psi^+)$.

\begin{thm} \label{thm:UEidanalytic} Consider any submodel of the \UE\ model 
with an analytic parameterization general enough to 
have the JC model as a limit.
Then for generic parameters the rooted topological tree $\psi^+$ is identifiable.
\end{thm}

\begin{proof} 
Let $$f:\Theta \to  UE(\psi^+)$$ denote the parameterization map for the submodel on tree $\psi^+$. Then $f(\Theta)$ cannot lie entirely in $UE(\phi^+)$ for any $\phi^+\ne \psi^+$, since, as shown in the previous proof, there are points from the JC model in the closure of  $f(\Theta)$ which are not in the closed set $UE(\phi^+)$. Thus the set $f^{-1}(UE(\psi^+) \cap UE(\phi^+))$ is a proper analytic subvariety of $\Theta$, and hence of measure zero in it. Since there are only finitely many $\phi^+$, for generic points in $\Theta$ the resulting distribution lies in the UE model for $\psi^+$ only.
\end{proof}

Note that the CK model, which is analytically parameterized, has the JC model as a limit, since after choosing a JC substitution process one can let the population size $N\to0^+$. This effectively ``turns off" the coalescent process, as small population sizes result in rapid coalescence.

Geometrically, the \UE\ model on a particular tree is a convex set, since it can be expressed as the solution set for a system of linear equations and inequalities. It immediately follows that mixtures of instances of the UE model on the same tree, whether 
defined by integrals such as typical rates-across-sites models (e.g., the ultrametric GTR+$\Gamma$ coalescent mixture 
of \cite{CK15}) or as sums (e.g., an ultrametric mixture of coalescent mixtures, as in \cite{ALR18}), or both, are 
also submodels of \UE\ on that tree. Provided the model has an analytic parameterization, as all these examples do, 
Theorem \ref{thm:UEidanalytic} then says that the tree topology is generically identifiable. Even in cases of mixtures 
which have so many numerical parameters that dimension arguments show they cannot all be individually identifiable, 
the species tree topology remains so. This is a potentially valuable observation, as a scheme designed for inference 
of a tree under the \UE\ model may avoid some issues of model misspecification that might arise with a  more 
standard approach of restricting to very simple models 
(e.g. constant population size) so that all numerical parameters are identifiable as well.

\medskip

The above theorems of course imply the weaker statement that for 
the \UE\ model (and many analytic submodels of the \UE\ model) on four or more taxa, the unrooted species tree topology is identifiable.
As SVDquartets is designed to infer unrooted 4-taxon trees, 
this gives hope that it might also be able to infer the unrooted tree 
topology for distributions from the more general \UE\ model.
For this to be possible, it is necessary to prove that the specific flattening matrices
considered in the SVDquartets method satisfy certain rank conditions,
the content of the next theorem.

Recall that if a $\kappa\times\kappa\times\kappa\times\kappa$ array $P$ has indices 
corresponding to taxa $a,b,c,d$, then the flattening  $\Flat_{ab|cd}(P)$ is a $\kappa^2\times \kappa^2$ 
matrix with row and column indices in $\kappa\times\kappa$ and  $((i,j),(k,l))$-entry $P(i,j,k,l)$. 

\begin{thm} \label{thm:rankUEflattenings}
For $P\in\text{\UE}_\kappa(\psi^+)$, and $ab|cd$ any unrooted quartet induced from 
the tree $\psi^+$, let $\tilde P=P|_{\{a,b,c,d\}}$ denote the marginalization to the taxa $a,b,c,d$.
Then for all such $P$, $\Flat_{ab|cd}(\tilde P)$ has rank  at most $\binom{\kappa+1}2$, while for 
generic $P$, $\Flat_{ac|bd}(\tilde P)$ and $\Flat_{ad|bc}(\tilde P)$ have rank $\kappa^2$.
\end{thm}

\begin{proof} Since $\psi^+|_{\{a,b,c,d\}}$ has at least one cherry,  assume one is formed by $a,b$. Then symmetry under 
exchanging the $a,b$ indices of $\tilde P$ shows that for each $1\le i<j\le \kappa$, the $(i,j)$ and $(j,i)$ rows of
$\Flat_{ab|cd}(\tilde P)$ are identical. Thus that flattening has at most $\kappa^2-\binom{\kappa}{2}=\binom{\kappa+1}2$ distinct rows, and 
its rank is at most $\binom{\kappa+1}2$.

We prove the second statement for $\Flat_{ac|bd}(\tilde P)$, noting that the argument for $\Flat_{ad|bc}(\tilde P)$ is similar. To show that for generic $P\in\text{\UE}_\kappa(\psi^+)$, 
$\Flat_{ac|bd}(\tilde P)$ has full rank, 
it suffices to construct a single $P$ for which this flattening matrix is full rank. 
To see that this is the case, consider the algebraic variety
$$V_{ac|bd} = \{ P \in \mathbb{R}^{\kappa^{|X|}}| \det(\Flat_{ac|bd}(\tilde P)) = 0\}.$$
This variety is defined by a single degree $\kappa^{2}$ polynomial and contains
all of the points $P$ for which $\Flat_{ac|bd}(\tilde P)$ is singular.
If there is a single point $P \in \UE_\kappa (\psi^+)$ for which $\Flat_{ac|bd}(\tilde P) \not = 0$, then the affine space $\UE_\kappa (\psi^+)$ is not contained in $V_{ac|bd}$. Thus, 
the intersection of $\UE_\kappa (\psi^+)$ with $V_{ac|bd}$ 
is a proper subvariety of $\UE_\kappa (\psi^+)$, and hence of measure zero within it.
Thus, generically, $\Flat_{ac|bd}(\tilde P)$ is full rank.

To construct such a probability distribution, assign any positive lengths to the edges of $\psi^+$ so that it 
becomes ultrametric, and consider the $\kappa$-state JC model on it (with no coalescent process). This 
leads to a distribution $P \in \text{\UE}_{\kappa}(\psi^+)$. Then $\tilde P$ arises from the Jukes-Cantor model 
on the induced rooted 4-taxon tree. Since the JC model is time reversible, $\tilde P$ is also obtained by 
rooting the quartet tree at the MRCA of $a$ and $b$, with non-identity JC Markov matrices on each of the 
5 edges of this rerooted tree. Let $M_a,M_b,M_c,M_d$ 
 denote the Markov matrices on the pendant edges 
and $M_{int}$ on the internal edge, so that
$F=(1/\kappa) M_{int}$ is the distribution of pairs of bases at the endpoints of the internal edge. 
Let  $N_{ac}=M_a\otimes M_c$ and $N_{bd}=M_b\otimes M_d$ denote the Kronkecker products.
Then, following the details of \cite[Section 4]{AR06},
the flattening matrix may be expressed as
$$\Flat_{ac|bd}(\tilde P)=N_{ac}^T D N_{bd},$$ 
where $D$ is a $\kappa^2\times\kappa^2$ diagonal matrix formed from the entries of $F$. 

Since $M_{int}$ is assumed to be a non-identity JC matrix, $F$ has no zero
entries, so $D$ has rank $\kappa^2$. Similarly, the JC transition matrices
 $M_a,M_b,M_c,M_d$ are non-singular, and since the Kronecker product of non-singular
 matrices is non-singular, so are $N_{ac}^T$ and $N_{bd}$. Thus
$\Flat_{ac|bd}$ generically has full rank. 
\end{proof}

The argument in this proof, that generically the ranks of ``wrong'' flattenings of quartet distributions are large, proceeded by constructing an element of the \UE\ model using a parameterized model in the absence of a coalescent process. However, just as was done in Theorem 3.2, we can extend the conclusion to analytic submodels of the \UE\ model, such as those incorporating the coalescent. For instance, since the CK model has the non-coalescent JC model as a limit,
this implies that there are points in the \CK\ model that 
are arbitrarily close to the point $P$ constructed in the proof, which therefore must also have rank $\kappa^2$ flattenings, 
as matrix rank is lower semicontinuous. We can thus obtain the following  generalization of a result from \cite{CK15}.

\begin{thm} \label{thm:svdid} Consider any submodel of the $\UE(\psi^+)$ model 
with an analytic parameterization general enough to 
have the JC model as a limit.
If $\psi^+$ displays the quartet $ab|cd$,
then for all distributions $P$ in the model,
with $\tilde P=P|_{\{a,b,c,d\}}$, $\Flat_{ab|cd}(\tilde P)$ has rank  at most $\binom{\kappa+1}2$, while for generic 
$P$, $\Flat_{ac|bd}(\tilde P)$ and $\Flat_{ad|bc}(\tilde P)$ have rank $\kappa^2$.
\end{thm}

We note that our proof of this theorem has avoided the explicit calculations and more intricate arguments that appear in 
\cite{CK15} while also establishing the result in a more general setting. This is possible because of our use of a tensor $P$ in the 
closure of the \CK\ model, but not in the \CK\ model, as well as adopting the viewpoint of \cite{AR06} on flattenings as matrix products.

\medskip

Using the two preceding theorems on identifiability, the statistical consistency of the SVDquartets method can be obtained. When Chifman and Kubatko \cite{CK15} proved essentially the same result on ranks of flattenings
for the \CK\ model, they highlighted it as an identifiability result, but did not explicitly make a claim of consistency. 
The consistency result for SVDquartets was then unambiguously stated and proved in this setting in \cite{WascherKubatko2020}, which also gave an analysis of the convergence rate.

Here we show that their argument for the consistency of SVDquartets applies more generally to 
site patterns generated under the \UE\ model, 
as well as many submodels.
In particular, it validates the consistency of
inference under models allowing mixtures of coalescent mixtures which may have different substitution processes 
across the genome, as described in \cite{ALR18}.

To be precise, we must first specify some method of quartet amalgamation $M$, which takes a collection of one quartet tree for each 4-taxon subset of $X$ and produces an unrooted topological tree on $X$. 
In order to establish consistency, we require that
if all quartet trees in the collection given to the method $M$ are displayed on a 
common tree $T$ on $X$, then $M$ returns $T$. 
Following \cite{Warnow18}, we say such a method is \emph{exact}
 while recognizing that for large sets $X$ one generally must use a 
heuristic method  $M'$ that seeks to approximate $M$.

\begin{thm}\label{thm:svdids}
The SVDquartets method, using an exact method to construct a tree from a collection of quartets,  
gives a statistically consistent unrooted species tree topology estimator 
for generic parameters under the \UE\ model, and under any submodel 
with an analytic parameterization 
general enough to 
have the JC model as a limit.
\end{thm}

\begin{proof} To simplify notation in the argument, let  $\Flat_{ac|bd}(P)$ denote the $ac|bd$ flattening of the marginalization  $P|_{\{a,b,c,d\}}$.

By Theorems  \ref{thm:rankUEflattenings} and \ref{thm:svdid} for generic parameters giving a probability distribution $P$ 
in the model and any  four taxa $a,b,c,d$ such that $ab|cd$ is displayed on the unrooted tree $\psi$, $\Flat_{ab|cd}(P)$ has 
rank  at most $\binom{\kappa+1}2$, while $\Flat_{ac|bd}(P)$ and $\Flat_{ad|bc}(P)$ have rank $\kappa^2$. 
This implies that $\Flat_{ab|cd}(P)$ will have at least  $\binom \kappa 2$ singular values of 0, while  
$\Flat_{ac|bd}(P)$ and $\Flat_{ad|bc}(P)$ have all positive singular values. For a finite sample of $s$ 
sites from the model, denote the empirical distribution by $\hat P_s$. Then for any $\epsilon>0$ and any norm
$$\lim_{s\to \infty}\operatorname{Pr}\big(|\hat P_s-P|<\epsilon \big) =1.$$
Since the vector $\sigma(M)$ of ordered singular values of a matrix $M$ is a  continuous function of the 
matrix, this implies that for each $q \in \{ab|cd, ac|bd, ad|bc\}$
$$\lim_{s\to \infty}\operatorname{Pr}\bigg( \| \sigma(\Flat_q(\hat P_s))-\sigma(\Flat_q(P))\|<\epsilon \bigg) =1$$
where $\| \cdot \|$ denotes any vector norm.
With the SVD score $\mu(M)$ defined as the sum of the  $\binom \kappa 2$ smallest singular values 
of a $\kappa^2\times\kappa^2$ matrix $M$, we know
$$0=\mu \big( \Flat_{ab|cd}(P) \big)<\min \bigg\{\mu (\Flat_{ac|bd}(P) ),\ \mu \big(\Flat_{ad|bc}(P) \big) \bigg\}.$$
But it then follows that  
$$\lim_{s\to \infty}\operatorname{Pr}\bigg( \mu(\Flat_{ab|cd}(\hat P_s) \big)<\min \bigg\{\mu(\Flat_{ac|bd}(\hat P_s)),\mu(\Flat_{ad|bc}(\hat P_s))\bigg\} \bigg) = 1.$$
Thus, as the sample size $s$ grows, the probability that choosing the quartet tree on $a,b,c,d$ 
minimizing $\mu$ gives the quartet tree displayed on $\psi$ approaches 1.

Since this probability approaches 1 for each of set of four taxa, and there are only finitely many such sets, 
the probability that all quartet trees inferred by minimizing $\mu$ are displayed on the species tree 
approaches 1. Thus with probability approaching 1, the method $M$ will return the correct species tree.
\end{proof}

\section{Dimension of \UE\ models on large trees}\label{sec:UEdim}

Although the symmetry conditions of the \UE\ model have been expressed as linear constraint equations, 
these constraints are not in general independent, as was shown for a particular 3-taxon species tree in Example \ref{ex:UEmodel}. 
In that example, it was 
easy to determine a basis of constraints, and thus the dimension of the model. 
In this section we investigate larger trees and determine the model dimension.

Knowledge of dimension is important for several reasons. First, it gives us a basic insight into how restrictive the model on a particular tree
topology is. Second, if one is to use these models for tree inference, the dimension is important for judging how close a data point is to fitting the model.
Intuitively, data is conceptualized as coming from a true model point with `noise' added, and if a model has high dimension the noise tends to do less to move that data from the model than if it had lower dimension. Such dimensionality considerations are made rigorous in many model selection criteria, for instance 
the Akaike Information Criterion and Bayesian Information Criterion.

\medskip

For a rooted topological tree $\psi^+$ on taxa $X$ we consider the model  \UE$_\kappa(\psi^+)$. Let ${d}_\kappa(\psi^+)$ 
denote the dimension of the affine space $V(\psi^+)\subset \RR^{\kappa^{|X|}}$ of all tensors satisfying the linear equations 
expressing the symmetry conditions defining the model, as well as that
all entries of the distribution tensor sum to 1 (i.e., the affine, or Zariski, closure of the model). 
By dropping the condition that tensor entries sum to 1, we pass to the 
cone over the model, a linear space $L(\psi^+)$ 
of dimension $c_\kappa(\psi^+)={d}_\kappa(\psi^+)+1$.  We now give a recursive formula for 
computing the dimension $c_\kappa(\psi^+)$.

\begin{thm}\label{thm:UEdim}
For a rooted binary topological tree $\psi^+$ on a taxon set $X$,
let $\psi_A^+$ and $\psi_B^+$ be the rooted subtrees descendant from the child nodes of the root of $\psi^+$, 
on taxa $A$ and $B$ respectively,  so that
$X=A\sqcup B$ and $\psi^+=(\psi_A^+,\psi_B^+)$.  Then 
$$c_\kappa(\psi^+) =
c_\kappa(\psi_A^+)c_\kappa(\psi_B^+)
 - \binom{\kappa}{2}.
$$
\end{thm}

\medskip

For a topological rooted species tree $\psi^+$ on $X$, 
we can construct a set of equations defining the cone $L(\psi^+)$ by considering every subset $Y \subseteq X$ and every 2-clade $\{a, b\}$ of each $\psi^+_{|Y}$ as was done in Example \ref{ex:UEmodel}.
However, as we saw in that example, the equations we obtain in this way are not necessarily independent.
As a first step towards proving Theorem \ref{thm:UEdim}, we construct a smaller (though still not necessarily independent) set of linear 
equations defining the cone $L(\psi^+)$. 
This set is defined by associating a set of linear equations to each vertex of the topological rooted
 tree $\psi^+$ on $X$.
Specifically, for each internal vertex $v$ of $\psi^+$ choose two taxa $a ,b$ with $v=\MRCA(a,b)$. 
Let  $P$  be a $|X|$-dimensional $\kappa \times\cdots\times \kappa$ tensor of indeterminates, with 
indices corresponding to taxa in $X$ and let $P_{a b }$ denote the marginalization of $P$ over 
all indices corresponding to taxa in $\desc(v)\setminus \{a ,b\}$. 
Each choice of the indices corresponding to taxa in $X \setminus \desc(v)$ determines a matrix 
slice of $P_{a b}$, with indices corresponding to $a ,b$. 
Expressing that each of these slices is symmetric yields a collection of linear equations.
Denote this set of equations by $\mathcal S_v=\mathcal S(\psi^+,\{a,b\})$.
Though the set $\mathcal S_v$ will depend on the particular pair of taxa $(a,b)$ chosen,
for our purposes the particular pair is irrelevant, so one can designate any consistent rule 
for selecting the pair $(a,b)$ so that the $\mathcal S_v$ are well-defined. 
If $v$ is not an internal vertex of $\psi^+$, define $S_v$ to be the empty set.

\begin{lemma}  \label{lem:smallDef} 
Let $\psi^+$ be a topological rooted tree on $X$. 
Then the set  
$$\mathcal{S} = \displaystyle \bigcup_{v \in V(\psi^+)} \mathcal{S}_v$$
defines the cone $L(\psi^+)$.
\end{lemma}

\begin{proof}
It is enough to show that if $v=\MRCA(a, b)=\MRCA(a,c)$, then the linear equations expressing symmetry 
of slices of $P_{ac}$ are contained in the span of those expressing symmetry of slices of $P_{ab}$ together 
with those equations in $\mathcal S$ arising from nodes descended from $v$. We show this inductively, proceeding 
from the leaves of the tree to the root. The base case, when $v$ has only two leaf descendants, is trivial. 
   Assume the result holds for the internal nodes descended from $v$. Let the children of $v$ be 
   $v_1$, which is ancestral to or equal to $a$, and $v_2$, 
   which is ancestral to $b,c$ since $\psi^+$ is binary. 
   Then $w=\MRCA(b,c)$ is a descendent of $v_2$. The equations arising from $w$ express that any 
   entry of the marginalization of $P$ over all descendants of $w$ except $b, c$ is invariant under 
   exchanging the $b, c$ indices. Since the entries of $P_{ab}$ arise from further marginalization, 
   the equations expressing symmetry of the $ab$-slices together with those arising from $w$ imply  
   those expressing the $ac$-slices of $P_{ac}$ are symmetric.
\end{proof}

The proof of the previous lemma explains the dependence of the equations we see in Example \ref{ex:UEmodel}. The $\{a,b,c\}$ constraints are the equations arising from MRCA$(b,c)$, which 
in that example, required no marginalization of $P$. The $\{a,b\}$ constraints are the equations arising from the root of the tree that express symmetry of $P_{ab}$ which are obtained by marginalizing $P$ over $c$. Together, these constraints imply the $\{b,c\}$ and $\{a,c\}$ constraints,
the latter of which express symmetry in the slices of $P_{ac}$.

\begin{proof}[Proof of Theorem \ref{thm:UEdim}]
Let $n_A=|A|$ and $n_B=|B|$.
With $U= \mathbb{R}^{\kappa^{n_A}}$ and $V=\mathbb{R}^{\kappa^{n_B}}$, we identify $W=U\otimes V=\mathbb{R}^{|X|}$ with the space
of $k^{n_A}\times k^{n_B}$ real matrices. 
In particular, we have $L(\psi_A^+)\subset U$,  $L(\psi_B^+)\subset V$, and $L(\psi^+)\subset W$.

We first claim that $L(\psi_A^+)\otimes L(\psi_B^+)$ is the subspace $Z\subset W$ defined by the subset $\mathcal S'$ of 
$\mathcal S = \mathcal S(\psi^+)$ of Lemma \ref{lem:smallDef} containing only those equations arising from non-root internal nodes of $\psi^+$.

To see $L(\psi_A^+)\otimes L(\psi_B^+)\subseteq Z$, consider an equation in $\mathcal S'$ associated to a non-root node $v$ and its 
descendant taxa $a,b$ as in the lemma. Without loss of generality, we may assume $v$ is a node of $\psi_A$. Then, ordering the 
taxa so that $a,b$ are the first two, this equation in $\mathcal S'$ has the form
\begin{equation}\label{eq:xx}
\sum_ {\alpha_1} x_{(i,j,\alpha_1,\alpha_2),\beta} -\sum_ {\alpha_1} x_{(j,i,\alpha_1,\alpha_2),\beta} =0
\end{equation}
where the summation over $\alpha_1\in [k]^m$ runs through all assignments of states to taxa descended from $v$ other than 
$a,b$ , $\alpha_2\in [k]^{n_A-2-m}$
 is a fixed choice of states for taxa in $A$ not descended from $v$, $\beta\in [k]^{n_B}$ is a
 fixed choice of states for the taxa in $B$, and $i\ne j$. This equation expresses that column $\beta$ of a matrix in $W$ satisfies 
 an equation associated to $v$, $a$, and $b$ in the definition of $L(\psi_A^+)$. Thus it holds on all of $L(\psi_A^+)\otimes L(\psi_B^+)$, and we 
 obtain the desired inclusion.

To see $L(\psi_A^+)\otimes L(\psi_B^+)\supseteq Z$, note that equation \eqref{eq:xx} has shown that every 
column of $z\in Z$ lies in $L(\psi_A^+)$, and likewise every row of $z$ lies in $L(\psi_B^+)$. But from the singular value decomposition of $z$, 
$$z=\sum_i c_i\otimes r_i$$
where the $c_i$ form a basis for the column space of $z$ and the $r_i$ form a basis for the row space of $z$. Since $c_i\in L(\psi_A^+)$ and
$r_i\in L(\psi_B^+)$, it follows that $z\in L(\psi_A^+)\otimes L(\psi_B^+)$, establishing the stated inclusion and that $Z = L(\psi_A^+)\otimes L(\psi_B^+)$.

Now the space $L(\psi^+)$ is the subset of $Z=L(\psi_A^+)\otimes L(\psi_B^+)$ defined by the equations in $\mathcal S\setminus \mathcal S'$, 
associated to the root of $\psi$. To conclude that
$$c_\kappa(\psi^+) =
c_\kappa(\psi_A^+)c_\kappa(\psi_B^+)
 - {\kappa \choose 2},
$$
it is enough to show that we can obtain an independent set of equations defining $L(\psi^+)$ by taking an independent set 
defining $Z$ and augmenting it by $\kappa \choose 2$ additional independent equations associated to the root. 

Let $\mathcal L$ be any independent subset of  equations in $\mathcal S'$ that define $Z$, and $\mathcal M=\mathcal S\setminus \mathcal S'$ 
the set of ${\kappa \choose 2}$ equations associated to 
the root of $\psi^+$ (and the choice of $a\in A$ and  $b\in B$).
Then $\mathcal L\cup \mathcal M$ defines $L(\psi^+)$. To see that $\mathcal L\cup \mathcal M$ is independent, first order indices so that 
$a$ and $b$ indices are  listed first among $A$ and $B$.
Then, using `+' in an index to denote the sum over the assignment of all states
$[\kappa] = \{1, 2, \dots, \kappa\}$ in that index, 
for any $1 \leq i<j \leq k$, $$x_{i+\ldots+,\; j+\ldots+} - x_{j+\ldots+,\; i+\ldots+}=0$$ 
must be the unique element of $\mathcal M$ that involves the variable $x_{ii\cdots i,\; jj\cdots j}$ (noting that each equation in $\mathcal L$ involves variables that 
have at least two distinct entries in the indices for $A$ or two distinct entries in the indices for $B$).
Since $\mathcal L$ is an independent set, this implies $\mathcal L\cup \mathcal M$ is independent.
\end{proof}

The theorem gives insight into model dimensions for families of `extreme' topologies:
rooted caterpillars and fully balanced shapes.

\begin{cor} \label{cor:catdim} Suppose $\psi^+$ is a rooted caterpillar tree on $n$ taxa. Then the 
dimension of the \UE$_\kappa(\psi^+)$ model  is
$${d}_\kappa(\psi^+)=\frac{\kappa^n+\kappa}2 -1.$$
\end{cor}

\begin{proof}
If $n=1$, then the model is simply a base distribution for the sole taxon, so $d_\kappa(\psi^+)=\kappa -1,$ 
consistent with the stated formula.
Now inductively assume the stated formula for the rooted caterpillar on $n-1$ taxa. Then by 
Theorem \ref{thm:UEdim}, for $n$ taxa
$$c_\kappa(\psi^+)=\left(\frac{\kappa^{n-1}+\kappa}2\right)\kappa -\binom \kappa2=
\left( \frac{\kappa^{n}+\kappa^2}2 \right) -\left(\frac {\kappa^2-\kappa}2 \right)=
\frac{\kappa^n+\kappa}2,
$$  and the claim follows.
\end{proof}

Also from Theorem \ref{thm:UEdim} we can compute that the dimension of the \UE\ model on the 4-taxon balanced tree $((a,b),(c,d))$ is
$$d_\kappa = \left ( \frac {\kappa^2+\kappa}{2}\right)^2-\binom \kappa 2-1=   \frac {\kappa(\kappa^3+2\kappa^2-\kappa+2)}4 -1.$$ 
By comparing the dimensions for the 4-taxon caterpillar and balanced trees, we see that $d_k$ depends on the rooted tree topology, and 
not only on the number of taxa.

More generally, for a fully balanced tree $\psi^+$ on $n=2^\ell$ taxa, Theorem \ref{thm:UEdim} yields that
$${d}_\kappa(\psi^+)=\co\left ( \left ( \frac {\kappa(\kappa+1)}2 \right )^{n/2} \right ).$$
Thus for fully balanced trees the dimension is $o(\kappa^n/2)$, while 
for rooted caterpillars on $n$ taxa, Corollary
\ref{cor:catdim} shows the dimension is asymptotic to $\kappa^n/2$. 
For a fixed number of taxa $n = 2^\ell$, it follows that 
the dimension of the balanced tree model will be smaller than that of the caterpillar.

This comparison of model dimension for caterpillars and balanced trees is intuitively plausible, as cherries on the 
full tree lead to more symmetry requirements on a tensor than do cherries on subtrees. In general, the more balanced a tree is, 
the smaller one might expect the model dimension to be. This leads us to pose the following
conjectures, where $RB(n)$ denotes the set of rooted binary $n$-leaf trees.

\begin{conj}
For all $\kappa$, there exists an $m$, such that for $n \geq m$, 
 ${d}_\kappa(\psi^+)$ is maximized over $\psi^+ \in RB(n)$ when $\psi^+$ is the $n$-leaf caterpillar tree. 
\end{conj}

\begin{conj}
For all $\kappa$, there exists an $m$, such that for $\ell \geq m$,  ${d}_\kappa(\psi^+)$ is minimized over $\psi^+ \in RB(2^\ell)$
when $\psi^+$ is the $2^\ell$-leaf balanced tree.
\end{conj}

\section{A genomic model of site patterns on general trees
}\label{sec:genmodel}

In this section, we examine a generalization of the CK model to non-ultrametric trees 
to motivate an algebraic model that encompasses it. 
Marginalizations (respectively, slices) of a site pattern probability tensor will be denoted by placing a `+' 
(resp. $k$) in the index summed over (resp. conditioned on). 
The transpose operator will be denoted with an exponent `$T$.'
For example, we can generalize the equations derived in Example \ref{ex:UEmodel} for the \UE\ model on the ultrametric 3-leaf rooted tree $(a, (b,c))$ for any value of $\kappa$ using this notation as follows:

\begin{center}

\begin{minipage}{4cm}
\begin{enumerate}
\item[(1)] \label{cond:symm} $P_{k \cdot \cdot}=P_{k \cdot \cdot}^T$, 
\item[(2)]  \label{cond:marg}
$P_{+\cdot \cdot}=P_{+\cdot\cdot}^T$.
\end{enumerate}
\end{minipage}
%
\begin{minipage}{4cm}
\begin{enumerate}
\item[(3)]  $P_{\cdot\cdot+}=P_{\cdot\cdot+}^T$,
\item[(4)]  $P_{\cdot+\cdot}=P_{\cdot+\cdot}^T$.
\end{enumerate}
\end{minipage}

\end{center}
In the definition of the \UE\ model, these constraints arise from the taxon subsets 
$(1) \, \{a,b,c\}$, $(2) \, \{b,c\}$, $(3)\,  \{a,b\}$ and $(4)\,  \{a,c\}$, and it is not hard
to see that the equations in (1) and (3) imply those in (2) and (4), just as in Example
\ref{ex:UEmodel}.

\subsection{The Extended Exchangeability Model}
\label{sec:EE}

In \cite{LK18}, the \CK\ model is extended to permit non-ultrametricity of the species tree. This extension allows, for instance, the modeling of relationships between species when generation times or scalar mutation rates differ across 
populations in the tree. 
In this same work, flattening matrices are used to establish the generic identifiability of the 
unrooted species tree topology of the extended model from which it follows  
that SVDquartets is still a statistically consistent 
method of inference of the unrooted species tree topology for these models when combined with any exact method of quartet amalgamation.

In order to motivate our algebraic model, 
first consider a model obtained from the \CK\ by dropping the ultrametricity requirement 
on the species tree. 
Suppose $a$ and $b$ are taxa in a 2-clade on $\sigma^+$, and let $v$ be their common parental node. In the 
special case that the edge lengths of $e_a=(v,a)$ and $e_b=(v,b)$ equal, then the lineages $a$ and $b$ would be 
exchangeable under this site pattern model as shown for the \CK\ model. 
Thus, for this particular tree the site pattern distribution 
can be viewed as a tensor with symmetry in the $a$ and $b$ coordinates. On a general species tree, however, 
where $e_a$ and $e_b$ may have different lengths and mutation rates may not be consistent, 
all sites evolve over those edges according to the transition matrices 
\begin{align*}M_a&=\exp\left (s_a Q\right ),\ \  s_a=\int_0^{\ell(e_a)} \mu_{e_a}(t)dt,\\
M_b&=\exp\left (s_b Q\right ),\ \  s_b=\int_0^{\ell(e_b)} \mu_{e_b}(t)dt,
\end{align*}
where $\ell(e)$ is the length of edge $e$ and $\mu_{e_a}(t)$ and $\mu_{e_b}(t)$
are time dependent mutations rates.

Supposing, without loss of generality, that $s_a\le s_b$, define the Markov matrix $$M=M_bM_a^{-1}=\exp\left ((s_b-s_a) Q\right ).$$
Then the site pattern distribution can be viewed as one obtained from a tensor with symmetry in $a$ and $b$ that has been acted on by $M$ in the $b$-index. 
More specifically, we imagine that on the edges leading toward both $a$ and $b$, the Markov matrix $M_a$ describes an initial substitution process, but on the edge to $b$ there is a subsequent mutation process described by $M$. If we introduce an additional action by $M$ on the edge to $a$, then in the resulting distribution we would regain symmetry in $a$ and $b$. Since no coalescent events occur in these pendant edges, there are no complications arising from the coalescent events that do occur.

To formalize this mathematically, suppose
$P$ is an $N$-way $\kappa\times\kappa\times\cdots \times \kappa$ tensor. 
Define the action of a $\kappa\times\kappa$ matrix $M$ in the $k$th index of $P$ by $Q=P*_k M$ where
$$Q(i_1,i_2,\ldots,i_k,\ldots, i_N)=vM,$$
with $v$ the row vector determined by fixing the $\ell$th index of $P$ to be $i_\ell$ for all $\ell\ne k$. 
For example, for $n = 3$ and $k=1$, the tensor $P*_1M$ is specified by 
$(P*_1M)_{ijk} = (P_{\cdot j k} M)_i$.
Given an $n$-tuple of matrices $(M_1,M_2,\dots M_n)$, let
$$P*(M_1,M_2,\dots, M_n)=(\dots ((P*_1M_1)*_2M_2)\dots *_nM_n)$$ denote the action in each of the indices of $P$.

\begin{defn}
\label{defn: ee}
Let $\psi^+$ be a rooted topological species tree on $X$ with $|X|=n$. Then the \emph{extended exchangeable model}, $\text{\EE}_\kappa(\psi^+)$, is the set of all $n$-way site pattern probability tensors $P$, 
such that there is an $n$-tuple $M=(M_1,M_2,\dots, M_n)$ of $\kappa\times \kappa$ non-singular Markov matrices $M_i$ and a non-negative array $\tilde P$ in the model $\text{\UE}_\kappa(\psi^+)$ such that $P*M=\tilde P$. 
\end{defn}

We note that  \UE\  is a submodel of \EE: any distribution in $\text{\UE}_\kappa(\psi^+)$ is seen to lie in $\text{\EE}_\kappa(\psi^+)$  by taking all matrices $M_i$ to be the identity. Also, to ensure that the \EE\ model does not include all distributions, it is important that the $M_i$ be non-singular in this definition:  Otherwise, if the $M_i$ describe processes where \emph{all} states transition to the same state with probability 1, then for any tensor
$P$, $P*(M_1,M_2,\dots M_n)=\tilde P$, a tensor with a single diagonal entry equal to 1 that is in UE.

\medskip

While the \UE\ model on a 2-leaf tree imposes constraints on the probability distribution of site patterns, the 2-leaf \EE\ model is dense among all probability distributions. Indeed, the \EE\ model on such a tree simply requires that the site pattern distribution have the form of $P= M_1^{-T}SM_2^{-1}$ with $S$  a  symmetric probability matrix and the $M_i$ Markov. But a dense subset of all probability distributions can be expressed as $P=DM$ for a  diagonal matrix $D$ with entries from the row sums of $P$ and an invertible Markov matrix $M$. We can thus take $M_1=M$, $S=M^TDM$, and $M_2=I$.

For a 3-taxon tree, though, the \EE\ model is typically not the full probability simplex $\Delta^{\kappa^3-1}$. 
For $\kappa \geq 4$, this follows from a simple dimension bound. The $UE(\psi^+)$ model for a 3-taxon tree $\psi^+$ has, from Corollary 4.3, dimension
$$d_\kappa = \frac {\kappa^3+\kappa}2 -1.$$
Moreover, the affine closure of the UE model on a 3-taxon tree is mapped to itself by the * action of $(M^{-1},M^{-1},M^{-1})$ for any Markov matrix $M$. Thus the dimension of the $EE(\psi^+)$ model can be at most $$\dim(UE(\psi^+)+2\kappa(\kappa-1),$$ where the second term is the number of parameters specifying two Markov matrices. Thus
$$\dim(EE(\psi^+)\le\frac{\kappa^3+4\kappa^2-3\kappa}2 -1< \kappa^3-1$$ for all $\kappa\ge 4$.

As we address in the remark following Corollary \ref{cor:EEinvariants}, 
we can confirm computationally that for a 3-taxon tree and $\kappa = 3$, $EE(\psi^+)$ is of lower dimension
than the probability simplex $\Delta^{26}$ and that for $\kappa = 2$,
the Zariski closure of $EE(\psi^+)$ is equal to $\Delta^{7}.$

\begin{rmk} A more restrictive variant of the \EE\ model, that is closer to the mechanistic models of \cite{LK18} model, could be defined by requiring that all the `extension' matrices $M_i$ arise as exponentials of the same GTR rate matrix.  While this common exponential condition is not expressible purely through algebra, there are other algebraic relaxations of it that one could impose instead, such as that the extension matrices $M_i$ are symmetric and commute.
\end{rmk}

\section{The \EE\ model on 3-taxon trees}\label{sec:EE3taxa}

By Definition \ref{defn: ee}, the \EE\ model on a 3-leaf rooted tree $\psi^+$ is the set of  $\kappa\times \kappa \times \kappa$ probability tensors of the form
$$P=\tilde P *(M_a^{-1},M_b^{-1},M_c^{-1}),$$ 
for some $\tilde P \in$ \UE$(\psi^+)$ and invertible Markov matrices $M_a,M_b,$ and $M_c$.

Because of the matrix actions, this model has a non-linear structure. This makes it more difficult to fully characterize the model \EE\ in terms of constraints than it was for the affine linear \UE\ model. It also means that the optimization problem for maximum likelihood may not be a convex one, making direct use of constraints for inference more appealing than attacking the optimization problem inherent to maximum likelihood. 

While determining all equality constraints satisfied by the model (i.e., generators of the ideal of model invariants) is difficult computationally, here we focus on determining some of them. We will use these in  Section \ref{sec:EEid} in our proof of tree identifiability under the \EE\ model.
Noting that only a few constraints are utilized in the SVDquartets method, future work should investigate whether the constraints found here are useful for rooted tree inference.

\begin{prop} \label{prop:CEEsym}Let $P$ be a tensor in the EE model on $\psi^+=((a,b),c)$, and $\Cof(A)$ denote the matrix of cofactors of a square matrix $A$.
Then for all $k \in [\kappa]$ the matrices 
$$Q^a_{\cdot\cdot k}=P_{+\cdot\cdot}\Cof(P_{\cdot+\cdot})^T P_{\cdot\cdot k}$$
and 
$$Q^b_{\cdot\cdot k}=P_{\cdot\cdot k}\Cof(P_{+\cdot\cdot})P_{\cdot+\cdot}^T$$
are symmetric: that is,
\begin{equation}
Q^a_{\cdot\cdot k}= ( Q^a_{\cdot\cdot k}) ^T \label{eq:EEleft}
\end{equation}
and
\begin{equation}
Q^b_{\cdot\cdot k}=( Q^b_{\cdot\cdot k})^T.\label{eq:EEright}
\end{equation}
\end{prop}

\begin{proof}
If $P$ is in the \EE\ model,  then $P=\tilde P *(M_a^{-1},M_b^{-1},M_c^{-1})$, with $\tilde P\in$ \UE\ and $M_a,M_b,M_c$ Markov.
Then $$P_{\cdot+\cdot}=M_a^{-T}\tilde P_{\cdot+\cdot} M_c^{-1}
\text{ \  and \  }
P_{+\cdot\cdot}=M_b^{-T}\tilde P_{+\cdot\cdot} M_c^{-1}=M_b^{-T} \tilde P_{\cdot+\cdot}M_c^{-1}$$
since $\tilde P\in$ \UE\  implies $\tilde P_{\cdot+\cdot} = \tilde P_{+\cdot\cdot}.$
Then, assuming necessary inverses exist, $$ P_{\cdot+\cdot}^{-T}P_{+\cdot\cdot}^{T}=M_aM_b^{-1},$$

Thus
$$P*_a(P_{\cdot+\cdot}^{-T}P_{+\cdot\cdot}^{T})=\tilde P*(M_b^{-1},M_b^{-1},M_c^{-1}).$$ 
But it is straightforward to check that every slice with fixed $c$-index of $\tilde P*(M_b^{-1},M_b^{-1},M_c^{-1})$ is symmetric, since that is true for $\tilde P$.
Thus $$(P_{\cdot+\cdot}^{-T}P_{+\cdot\cdot}^{T})^T P_{\cdot\cdot k}=P_{+\cdot\cdot} P_{\cdot+\cdot}^{-1} P_{\cdot\cdot k}$$
is symmetric for every $k$.
Using the cofactor formula for the inverse of a matrix, and clearing denominators by multiplying by a determinant yields \eqref{eq:EEleft}.

The assumption of invertibility used in this argument can be justified for a dense set of choices of $\tilde P$. Indeed, it is enough to exhibit
one such choice, since that indicates the  subset leading to non-invertibility is a proper subvariety (defined by certain minors vanishing), and hence of lower dimension. Such a choice is obtained with the Markov matrices being the identity, and $\tilde P$ having non-zero diagonal entries, and zero elsewhere.
Since the claim is established on a dense set, it holds everywhere by continuity.

The claim \eqref{eq:EEright} can be shown either in a similar way, or by conjugating (in the sense of multiplying by a matrix and its transpose) $Q^a_{\cdot\cdot k}$ by $P_{\cdot + \cdot} P_{+\cdot\cdot}^{-1}$  and removing determinant factors.
\end{proof}

\begin{rmk} Since $Q^a_{\cdot\cdot k}$ and $Q^b_{\cdot\cdot k}$ are conjugate for any tensor $P$ (even one not in the \EE\ model), checking that one is symmetric implies the other is as well, provided the appropriate inverse exists. If these are used as necessary conditions for membership in the model, when applied to data it may still be desirable to check that both are approximately symmetric, since it is unclear how conjugation will effect the way we measure the inevitable stochastic error leading to violation of perfect symmetry.
\end{rmk}

\begin{cor} 
\label{cor:EEinvariants}
The \EE\ model on $\psi^+=((a,b),c)$ is contained in the algebraic variety defined by the degree $\kappa+1$ polynomials given by the entries of the $2\kappa$ matrix equations
\begin{align*}
P_{+\cdot\cdot}\Cof(P_{\cdot+\cdot})^T P_{\cdot\cdot k} &- P_{\cdot\cdot k}^T \Cof(P_{\cdot+\cdot}) P_{+\cdot\cdot}^T,\\
P_{\cdot\cdot k}\Cof(P_{+\cdot\cdot}) P_{\cdot+\cdot}^T &-P_{\cdot+\cdot} \Cof(P_{+\cdot\cdot})^T P_{\cdot\cdot k}^T.
\end{align*}
\end{cor}

The polynomials of this corollary also arise as phylogenetic invariants for the general Markov (GM) model of sequence evolution \cite{AR03} with no
coalescent process. 
In the setting of that work, the tensors of interest are those in the orbits of 3-way diagonal tensors under actions of $GL_\kappa$ in each index, while here they are the orbits of tensors symmetric in two indices under the same $GL_\kappa$ actions. Since diagonal tensors display this symmetry, the invariants above must also apply to the GM model.
However, the GM model on a 3-taxon tree has additional invariants of this form, for every pair of taxa, not just those in the cherry.

\begin{rmk}
Using the computational algebra software Singular \cite{DGPS}, 
we are able to show that for $\kappa = 2$,
there are no non-trivial polynomials vanishing on the \EE\ model. Thus,
the polynomial invariants implied by Corollary \ref{cor:EEinvariants} are identically zero. 
For $\kappa = 3$, we verified computationally that these invariants are not identically zero.

\end{rmk}

As demonstrated by methods such as SVDquartets, reframing model constraints in terms of rank conditions can
be useful for developing practical methods of phylogenetic inference. 
With this in mind, we can reinterpret the results of 
Corollary \ref{cor:EEinvariants} as rank conditions for the \EE\ model. 
To do so, we use the following lemma, which follows a construction of G. Ottaviani that was suggested to us by L. Oeding.

\begin{lemma}\label{lem:Ott} Let $A,B,C,D,E, F$ be six $\kappa\times\kappa$ matrices, with $B,E$ invertible, satisfying $$CB^{-1}A+DE^{-1}F=0.$$
Then the $3\kappa\times 3\kappa$ matrix
$$\begin{pmatrix}0&A&B\\D&0&C\\E&F&0\end{pmatrix}$$
has rank $2\kappa$.
\end{lemma}

\begin{proof}
Observe
$$\begin{pmatrix}0&A&B\\D&0&C\\E&F&0\end{pmatrix}=\begin{pmatrix}I&0&0\\0&I&D\\0&0&E\end{pmatrix}\begin{pmatrix}0&0&I\\0&-(CB^{-1}A+DE^{-1}F)&CB^{-1} \\I&E^{-1}F&0\end{pmatrix}\begin{pmatrix}I&0&0\\0&I&0\\0&A&B\end{pmatrix}.$$
\end{proof}

\begin{cor}\label{cor:higherdegree} Tensors in the \EE\ model on $\psi^+=((a,b),c)$ are contained in the algebraic variety defined by the degree $2\kappa+1$ polynomials given by the  $(2\kappa+1)\times (2\kappa+1)$ minors of each of the $2\kappa$ matrices
$$\begin{pmatrix}0&P_{\cdot\cdot k}&P_{\cdot+\cdot}\\-P_{\cdot\cdot k}^T&0&P_{+\cdot\cdot}\\-P_{\cdot+\cdot}^T&-P_{+\cdot\cdot}^T&0\end{pmatrix}$$
and
$$\begin{pmatrix}0&P_{\cdot + \cdot }^T&P_{+\cdot\cdot}^T\\-P_{\cdot+ \cdot }&0&P_{\cdot\cdot k}\\-P_{+\cdot \cdot}&-P_{\cdot\cdot k}^T&0\end{pmatrix}.$$
\end{cor}
\begin{proof} Choosing $A,B,C,D,E,F$ in Lemma \ref{lem:Ott} as shown in these matrices makes the equation $CB^{-1}A+DE^{-1}F=0$ express that $Q^a_{\cdot\cdot k}$ and  $Q^b_{\cdot\cdot k}$ are symmetric, which was shown in Proposition \ref{prop:CEEsym}.
\end{proof}

The result of Corollary \ref{cor:higherdegree} allows one to formulate necessary conditions for EE model membership on the 3-taxon tree in terms of rank conditions on matrices, much as the SVDquartets method is based on rank conditions on matrices
in the 4-taxon case.


\section{Tree identifiability under the \EE\ model}\label{sec:EEid}
The \EE\ model invariants of the previous section enable us to prove that the rooted tree topology is generically identifiable under the \EE\ model. 
We establish these results for $\kappa \geq 4$, which includes the cases most relevant
for phylogenetic analysis.

To establish identifiability, we use the following \emph{non}-identifiability result.

\begin{lemma} \label{lem:2nonid} 
Consider a 2-taxon species tree $(a\tc x,b\tc (\ell-x) )$, with $0\le x\le \ell$ with constant population size $N$ above the root and any GTR rate matrix $Q$ with stationary distribution $\pi$. Then the probability distribution matrix $F$ of site patterns under the \CK\ model is symmetric and independent of $x$.
\end{lemma}
\begin{proof} Using time reversibility, the distribution can be expressed as
$$F=\int _{t=0}^\infty \operatorname{diag}(\pi) M_x M_{2t} M_{\ell-x} \mu_N(t) dt$$
where $\mu_N(t)$ is the density function for coalescent times,  and $M_z=\exp(Qz).$ Since the integrand, a GTR distribution, is a symmetric matrix, then so is $F$.
Since the $M_z$ commute, and $M_xM_{\ell-x}=M_\ell$,
$$F= \operatorname{diag}(\pi) M_\ell \int _{t=0}^\infty  M_{2t} \mu_N(t) dt$$
has no dependence on $x$.
\end{proof}

\begin{thm} \label{thm:EEid} The rooted topological tree $\psi^+$ is identifiable from generic probability distributions in the \EE$_\kappa(\psi^+)$ model for $\kappa\ge 4$.
\end{thm} 
\begin{proof}

We first suppose $\kappa=4$. For the 3-taxon trees $\phi^+=((a,b),c)$ and $\psi^+=((a,c),b)$, we show that \EE$(\psi^+)\cap$\EE$(\phi^+)$ has measure zero within \EE$(\psi^+)$. 
To do this, it is enough to construct one point in \EE$(\psi^+)$ that is not in the Zariski closure of \EE$(\phi^+)$, since that implies the Zariski closure 
of the intersection of \EE$(\psi^+)$ and \EE$(\phi^+)$ is of lower dimension than \EE$(\psi^+)$.

Let $N$ be an arbitrary effective population size and let
$\phi^+ = ((a\tc 2, c\tc 0)\tc 1,b\tc 1)$, with distances in coalescent units (number of generations divided by $2N$). 
Let $\mu = 1/2N$ and define $Q$ to be the Kimura 2-parameter (K2P) rate matrix
$$
\begin{pmatrix}
-4 & 1 & 2 & 1  \\
1 & -4 & 1 & 2 \\
2 & 1 & -4 & 1 \\
1 & 2 & 1 & -4 \\
\end{pmatrix}
$$
with equilibrium distribution $\pi =(\frac{1}{4}, \frac{1}{4}, \frac{1}{4}, \frac{1}{4})$.
Finally, let $P$ be the probability tensor that arises from this choice of
parameters in the \CK\ model.

Then letting $M=\exp(2Q)$, we see that $\widetilde P=P*(I,M,M)$ lies in $\UE(\psi^+)$, which implies that $P\in  \EE(\psi^+)$. 
To see that $P$ does
not belong to the Zariski closure of $EE(\phi^+)$ by Corollary \ref{cor:EEinvariants} it suffices to show that for some $k$
\begin{equation}
P_{+\cdot\cdot}\Cof(P_{\cdot+\cdot})^T P_{\cdot\cdot k} 
- P_{\cdot\cdot k}^T \Cof(P_{\cdot+\cdot}) P_{+\cdot\cdot}^T \ne 0. \label{eq:K2Pgenericpoint}
\end{equation}

Note that
$P_{+\cdot\cdot}$ and
$P_{\cdot+\cdot}$ are probability distribution matrices
for the same model on the 2-leaf species trees $(b\tc 1, c\tc 1)$ and $(a\tc 2, c\tc0)$.
But, by Lemma \ref{lem:2nonid},  
$$P_{+\cdot\cdot} = P_{\cdot+\cdot} = 
P_{+\cdot\cdot}^T = P_{\cdot+\cdot}^T,$$
so
$$
P_{+\cdot\cdot}P_{\cdot+\cdot}^{-T} = 
P_{\cdot+\cdot}^{-1} P_{+\cdot\cdot}^T  = 
I_4.$$
To show \eqref{eq:K2Pgenericpoint}, it is thus enough to show that some
$P_{..k}$ is not symmetric. This can be verified without appealing to numerical computation: 
For example, 
$$(P_{..1})_{12} - (P_{..1})_{21}=
\dfrac{1}{10530}e^{-20}-
\dfrac{1}{22230}e^{-25}-
\dfrac{1}{20007}e^{-29}.
$$
If this were zero, then multiplying by $e^{29}$ would show $e$ is a root of a rational polynomial, 
contradicting its transcendence. 

Thus \EE$(\psi^+)\cap$\EE$(\phi^+)$ has measure zero within \EE$(\psi^+)$.

Interchanging taxon names then shows the intersection of any two resolved 3-taxon tree models is of measure zero within them, and thus that a generic distribution in any single 3-taxon model lies only 
in that 3-taxon model. This establishes the theorem for 3-taxon trees when $\kappa=4$.

\smallskip

For larger trees $\psi^+$, each displayed rooted triple determines a measure zero subset of \EE$(\psi^+)$ containing all points where that rooted triple may not be identifiable from marginalizations of $P$ to those 3 taxa. Since there are a finite number of such sets, for a generic $P\in$ \EE$ (\psi^+)$, all displayed rooted triples are identifiable, and hence so is the tree $\psi^+$.

\smallskip

For $\kappa>4$, the proof follows by embedding 
the 4-state rate matrix above in the upper left corner of a 
$\kappa$-state GTR rate matrix and setting the remaining entries to 0.
\end{proof}

\begin{rmk}
Several comments are in order about the method of proof in Theorem \ref{thm:EEid}. First, if the matrix $Q$ is chosen to be  a Jukes-Cantor rate-matrix, then one finds that the same construction of $P$ leads to a point on which the invariants for \EE$(\phi^+)$ vanish. That is, $P$ is not `sufficiently generic' to identify the rooted tree. This is explored more thoroughly in the Appendix. 
\end{rmk}

Second, since the argument used an instance of the \CK\ model with a K2P rate matrix, it also establishes the following, which directly applies to models used for phylogenetic inference.

\begin{cor}
\label{cor: K2Psubmodel}
For $\kappa =4$, consider any submodel of \EE\ such that each $\psi^+$ has an analytic parameterization general enough to contain the Kimura 2-parameter coalescent mixture model with constant population size. Then for generic parameters the rooted topological tree $\psi^+$ is identifiable.
\end{cor}

Finally, while our proof of identifiability of a rooted tree under the \EE\ model fails for the \CK\ Jukes-Cantor model, unrooted trees are still identifiable under that model. To establish this,
note that a probability distribution for a 4-taxon tree on taxa $a,b,c,d$ under the  \EE\ model has the form $P=\tilde P*(M_a,M_b,M_c,M_d)$,
with $\tilde P$ in the UE model and the Markov matrices invertible. As a result, its flattenings can be expressed as
\begin{align*}
\Flat_{ab|cd}(P)&=(M_a\otimes M_b)^T \Flat_{ab|cd}(\tilde P) (M_c\otimes M_d),\\
\Flat_{ac|bd}(P)&=(M_a\otimes M_c)^T \Flat_{ac|bd}(\tilde P) (M_b\otimes M_d),\\
\Flat_{ad|bc}(P)&=(M_a\otimes M_d)^T \Flat_{ad|bc}(\tilde P) (M_b\otimes M_c).\\
\end{align*}

Since $M_a,M_b,M_c,$ and $M_d$ have full rank, this implies the rank of each flattening of $P$ is equal to the rank of the corresponding flattening of $\tilde P$.
It is then straightforward to obtain the following analog of Theorem \ref{thm:svdids}.

\begin{thm}\label{thm:svdidsEE}
The SVDquartets method, using an exact method to construct a tree from
a collection of quartets, gives a statistically consistent unrooted species tree topology estimator
for generic parameters under the EE model, and under any submodel with an analytic
parameterization general enough to contain the CK K2P model.
\end{thm}


\appendix
\section{Pseudo-exchangeability for the Jukes-Cantor Model}
\label{sec: pseudoexchangeability}

The proof of Theorem \ref{thm:EEid}, on the generic identifiability of the tree topology under the \EE\ model, used a particular point in the \EE\ model arising from the \CK\ Kimura 2-parameter model. Here, we show that it is not possible to use similar arguments with a point in the \CK\ Jukes-Cantor model. We do this by specifically considering the \CK\ Jukes-Cantor model, and showing that the model always has `extra symmetries' that prevent the identification of the rooted triple tree by these invariants.

\begin{prop} 
\label{prop: psuedo-exchangeability} Consider the \CK\ Jukes-Cantor model on the tree
$((a\tc \ell_a,b\tc \ell_b)\tc  \ell_{ab}, c \tc  \ell_c )$.
If $\ell_a = \ell_{ab} + \ell_c$ then the resulting probability tensor $P = (p_{ijk})$ exhibits $a,c$ exchangeability, that is,
$p_{ijk} = p_{kji}$. 
\end{prop}

\begin{proof}
Let $P = (p_{ijk})$ be a probability tensor from the \CK\ Jukes-Cantor model on a 3-leaf tree.
While $P$ has 64 entries,
%
because the site substitution model is the Jukes-Cantor model, it
 has at most five distinct entries.
Thus, we may group the coordinates of $P$ into five equivalance classes,
which we represent by
$$[p_{AAA}], [p_{AAC}],[p_{ACA}],[p_{ACC}],[p_{ACG}].$$
For any representative of the equivalence class
$[p_{AAA}]$, $[p_{ACA}],$ or $[p_{ACG}]$, 
swapping the first and third indices produces another representative
of the same equivalence class. 
However, for representatives of the equivalence class
$[p_{AAC}]$, swapping the first and third indices produces a 
representative of the equivalence class $[p_{ACC}]$,
and vice versa. Therefore, to prove the proposition, it 
suffices to show that for $P$, $[p_{AAC}]$ 
and $[p_{ACC}]$ are equal. To establish this, 
we prove that $p_{AAC} = p_{ACC}$.

Restricting to the leaf set $\{a,b\}$, we obtain the 
2-leaf rooted tree
$(a\tc \ell_a,b\tc \ell_b)$ and the probability of observing
state $ij$ from the \CK\ Jukes-Cantor model on this tree is
$$P_{ij+} = p_{ijA} + p_{ijC} + p_{ijG} + p_{ijT}.$$
Likewise, by restricting to the leaf set $\{b,c\}$, we obtain the
2-leaf rooted tree
$(b\tc \ell_b + \ell_{ab},c \tc \ell_c)$ and the probability of observing
state $jk$ from the \CK\ Jukes-Cantor model on this tree is
$$P_{+jk} = p_{Ajk} + p_{Cjk} + p_{Gjk} + p_{Tjk}.$$

Note that since $\ell_a = \ell_{ab} + \ell_c,$ 
the 2-leaf species trees obtained by restricting to $\{a,b\}$ and $\{b,c\}$ 
differ only by the location of the root.
By Lemma \ref{lem:2nonid}, since the Jukes-Cantor model is a
submodel of GTR, the probability distribution matrices for the
JC models on these trees are symmetric and equal.
Therefore, we have 
$P_{ij+} = P_{ji+}=  P_{+ji} = P_{+ij}$.
Specifically, this implies $P_{AC+} = P_{+CA}$, or
$$p_{ACA} + p_{ACC} + p_{ACG} + p_{ACT} = 
p_{ACA} + p_{CCA} + p_{GCA} + p_{TCA}.$$
Under the JC model, $p_{ACG}$, $p_{ACT}$, $p_{GCA}$, and $p_{TCA}$ 
all belong to the equivalence class of coordinates with 
three distinct indices, which is to say, 
$p_{ACG} = p_{ACT} = p_{GCA} = p_{TCA}$.
Thus, by cancellation, the equation above reduces to 
$p_{CCA} = p_{ACC}$. Since $p_{CCA}$ and $p_{AAC}$ are in the same
JC equivalence class, this implies
$p_{AAC} = p_{ACC}.$
\end{proof}
\begin{cor}
The invariants of Corollory \ref{cor:EEinvariants} associated to all of the trees $((a,b),c)$, $((a,c),b)$ and $((b,c),a)$ vanish on all probability tensors $P$ arising from the Jukes-Cantor \CK\ model on any of these trees.
\end{cor}
\begin{proof}
First consider $\tilde P$ arising from the \CK\ Jukes-Cantor model on the tree  $((a\tc \ell,b\tc \ell)\tc  \ell, c \tc  0 )$. By the proposition, this tensor is fully-symmetric, that is, invariant under any permutation of the indices, for any positive value of $\ell$. It thus lies in the \UE\ model for all three trees. Now
the probability tensor $P$ from the \CK\ JC model on 
$((a\tc \ell_a,b\tc \ell_b)\tc  \ell, c \tc  \ell_c )$, 
where $\ell_a,\ell_b\ge \ell$ and $\ell_c\ge 0$ 
can be expressed as
$$P=\tilde P*(M_a,M_b,M_c),$$
where $M_a$, $M_b$, $M_c$ are Jukes-Cantor matrices for edges of length $\ell_a-\ell$, $\ell_b-\ell$, $\ell_c-\ell$, respectively.  Thus $P$ lies in the \EE\ model for all three trees. Therefore the invariants associated to all three trees vanish on it.

Moreover, since the entries of probability tensors in the \EE\ model are parametrized by analytic functions of the edge lengths, composing these function with the invariants gives analytic functions that vanish on a full-dimensional subset of the parameter space, which must therefore be zero on the entire parameter space. Thus the invariants vanish  on the model even when the terminal edge lengths do not satisfy the assumed inequalities.
\end{proof}


\bibliographystyle{alpha}
\bibliography{ExchMod}

\end{document}